\documentclass[a4paper,11pt]{article}
\usepackage{authblk}
\usepackage{a4wide}
\pdfoutput=1

\usepackage[T1]{fontenc}   
\usepackage[latin1]{inputenc}
\usepackage[english]{babel} 
\usepackage{amsmath}
\usepackage{amsfonts}
\usepackage{amssymb}
\usepackage{amsthm}
\usepackage{mathrsfs}
\usepackage{stmaryrd}
\usepackage{url}
\usepackage{enumerate}
\usepackage{graphicx}
\usepackage[all]{xy}
\usepackage{algorithm}
\usepackage{algorithmic}
\usepackage{color}
\usepackage{hyperref}
\usepackage{booktabs}
\usepackage{marvosym}

\theoremstyle{plain}
\newtheorem{thm}{Theorem}

\newtheorem{assumption}[thm]{Assumption}{\bfseries}{\itshape}
{\bfseries}{\itshape}
\newtheorem{prop}[thm]{Proposition}

\newtheorem{lem}[thm]{Lemma}
\newtheorem{lemma}[thm]{Lemma}
\theoremstyle{definition}
\newtheorem{definition}[thm]{Definition}

\newtheorem{nota}[thm]{Notation}
\theoremstyle{remark}
\newtheorem{rem}[thm]{Remark}
\newtheorem{remark}[thm]{Remark}

\newcommand{\fract}[2]{\hbox{\leavevmode
\kern.1em \raise .3ex \hbox{\the\scriptfont0 $#1$}\kern-.1em }\big/
\raise -.5ex\hbox{\kern-.15em \lower .25ex \hbox{\the\scriptfont0 $#2$}}
}

\newcommand{\eqdef}{\stackrel{\text{def}}{=}}

\renewcommand{\leq}{\leqslant} 
 
\renewcommand{\geq}{\geqslant}

\newcommand{\Fq}{\ensuremath{\mathbb{F}_q}}

\newcommand{\AC}{\code{A}}
\newcommand{\BC}{\code{B}}
\newcommand{\CC}{\code{C}}

\newcommand{\CCi}{{\code{C}^{(i)}}}
\newcommand{\Gmi}{\Gm^{(i)}}

\newcommand{\mat}[1]{\ensuremath{\boldsymbol{#1}}}
\newcommand{\code}[1]{\ensuremath{\mathscr{#1}}}

\newcommand{\Am}{\mat{A}}

\newcommand{\Gm}{\mat{G}}

\newcommand{\Pm}{\mat{P}}

\newcommand{\Sm}{\mat{S}}

\newcommand{\av}{\mat{a}}
\newcommand{\bv}{\mat{b}}
\newcommand{\cv}{\mat{c}}

\newcommand{\ev}{\mat{e}}

\newcommand{\mv}{\mat{m}}

\newcommand{\vv}{\mat{v}}

\newcommand{\xv}{{\mat{x}}}

\newcommand{\yv}{{\mat{y}}}

\newcommand{\zerom}[1]{{\mat{0}_{#1}}}

\newcommand{\F}{\ensuremath{\mathbb{F}}}

\renewcommand{\Im}{\mat{I}}

\newcommand{\GRS}[3]{\text{\bf GRS}_{#1}(#2,#3)}

\newcommand{\sh}[2]{\mathcal{S}_{#2}\left(#1\right)}

\newcommand{\pu}[2]{\mathcal{P}_{#2}\left(#1\right)}
\newcommand{\re}[2]{\mathcal{R}_{#2}\left(#1\right)}

\newcommand{\sq}[1]{#1^{\star 2}}
\newcommand{\sqc}[1]{#1^{\star 2}} 
\newcommand{\sqb}[1]{\left(#1\right)^{\star 2}}

\newcommand{\kGRS}{k_{\rm GRS}}
\newcommand{\GSCGRS}{\Gm_{\rm SCGRS}}
\newcommand{\GRand}{\Gm_{\rm rand}}

\newcommand{\T}{\mathcal{T}}
\newcommand{\IGRS}{\mathcal{I}_{\rm GRS}}
\newcommand{\IPR}{\mathcal{I}_{\rm PR}}
\newcommand{\IR}{\mathcal{I}_{\rm R}}
\newcommand{\IT}{\mathcal{I}_{\rm T}}

\newcommand{\IInt}[2]{\llbracket #1, #2 \rrbracket}
\newcommand{\Pos}{\IInt{1}{n+w}}

\newcommand{\Jind}{\mathcal{L}}

\begin{document}

\title{Recovering short secret keys of RLCE\\in polynomial time}

\author[1]{Alain Couvreur\thanks{\texttt{alain.couvreur@lix.polytechnique.fr}}}
\affil[1]{Inria \&  LIX, CNRS UMR 7161\break
  \'Ecole polytechnique, 91128 Palaiseau Cedex, France.
}
\author[2,3]{Matthieu Lequesne\thanks{\tt matthieu.lequesne@inria.fr}}
\author[2]{Jean-Pierre Tillich\thanks{\tt jean-pierre.tillich@inria.fr}}
\affil[2]{Inria,   
  2 rue Simone Iff, 75012 Paris, France.
}
\affil[3]{Sorbonne Universit\'e, UPMC Univ Paris 06}

\maketitle

\begin{abstract}
  We present a key recovery attack against Y. Wang's Random Linear
  Code Encryption (RLCE) scheme recently submitted to the
  NIST call for post-quantum cryptography.  This attack recovers the secret key
  for all the short key parameters proposed by the author.
\end{abstract}


\section*{Introduction}
The McEliece encryption scheme dates back to the late 70's \cite{M78}
and lies among the possible post-quantum alternatives to number theory
based schemes using integer factorisation or discrete logarithm.
However, the main drawback of McEliece's original scheme is the large size of
its keys.  Indeed, the classic instantiation of McEliece using binary
Goppa codes requires public keys of several hundreds of kilobytes to
assert a security of 128 bits. For example, the recent NIST
submission {\em Classic McEliece} \cite{BCLMNPPSSSW17}
proposes public keys of $1.1$ to $1.3$ megabytes to assert 256
bits security (with a classical computer).

For this reason, there is a recurrent temptation consisting in using codes 
with a higher decoding capacity for encryption in order to reduce the size 
of the public key. Many proposals
in the last decades involve generalised Reed Solomon (GRS) codes, which
are well--known to have a large minimum distance together with
efficient decoding algorithms correcting up to half the minimum
distance. On the other hand, the raw use of GRS codes has been proved
to be insecure by Sidelnikov and Shestakov \cite{SS92}. Subsequently,
some variations have been proposed as a counter-measure of Sidelnikov
and Shestakov's attack. Berger and Loidreau \cite{BL02} suggested to
replace a GRS code by a random subcode of small codimension,
Wieschebrink \cite{W06} proposed to join
random columns in a generator
matrix of a GRS code and Baldi \textit{et al.} \cite{BBCRS16} suggested to
mask the structure of the code by right multiplying a generator matrix
of a GRS code by the sum of a low rank matrix and a sparse matrix. It
turns out that all of these proposals have been subject to efficient
polynomial time attacks \cite{W10, CGGOT14, COTG15}.

A more recent proposal by Yongge Wang \cite{W16} suggests another
way of hiding the structure of GRS codes.
The outline of Wang's construction is the following: start from a 
$k \times n$ generator matrix of a GRS code of length $n$ and dimension $k$
over a field $\F_q$, add $w$ additional random columns to the matrix,
and mix the columns in a particular manner. The design of this scheme is 
detailed in \S~\ref{subsec:presentation}.
This
approach entails a significant expansion of the public key size but
may resist above-mentioned attacks such as distinguisher and filtration
attacks \cite{CGGOT14, COT17}. This
public key encryption primitive is the core of Wang's recent NIST
submission ``RLCE--KEM'' \cite{W17}.

\paragraph{Our contribution} In the present article we give a
polynomial time key recovery attack against RLCE which breaks the
system when the number of additional random columns $w$ is strictly less
than $n-k$. 
This allows us to break half the parameter sets proposed in \cite{W17}.

\section{Notation and prerequisites}

\subsection{Generalised Reed--Solomon codes}

\begin{nota}
Let $q$ be a power of prime and $k$ a positive integer. We denote by $\F_q[X]_{<k}$ the vector space of polynomials over $\F_q$ whose degree is strictly bounded from above by $k$. 
\end{nota}

\begin{definition}[Generalised Reed Solomon codes]\label{def:GRS}
Let $\xv \in \F_q^n$ be a vector whose entries are pairwise distinct
and $\yv \in \F_q^n$ be a vector whose entries are all nonzero. The
 {\em generalised Reed--Solomon (GRS) code with support $\xv$ and
   multiplier $\yv$ of dimension $k$} is defined as
\[
\GRS{k}{\xv}{\yv} \eqdef \left\{(y_1 f(x_1), \ldots, y_n f(x_n)) ~|~
f \in \F_q[x]_{<k}\right\}.
\]
\end{definition}

\subsection{Schur product of codes and square codes distinguisher}

\begin{nota}
  The component-wise product of two vectors $\av$ and $\bv$ in $\F_q^n$ is denoted by
\[
\av \star \bv \eqdef (a_1b_1, \ldots, a_n b_n).
\]
This definition extends to the product of codes where the {\em Schur product}
of two codes $\code A$ and $\code B \subseteq \F_q^n$ is defined as
\[
\code A \star \code B \eqdef \mathbf{Span}_{\F_q} \left\{
  \av \star \bv ~|~ \av \in \code A, \ \bv \in \code B
  \right\}.
\]
In particular, $\sq{\code A}$ denotes the \emph{square code} of a code $\code A$: $\sq{\code A}\eqdef \code A \star \code A$.
\end{nota}

We recall the following result on the generic behaviour of random codes
with respect to this operation. 

\begin{prop}({\cite[Theorem 2.3]{CCMZ15}}, informal)\label{prop:CCMZ}
For a  linear code $\code R$ chosen at random over $\Fq$ of dimension $k$ and length $n$,
the dimension of $\code R^{\star 2}$ is typically $\min (n, {k+1 \choose 2})$. 
\end{prop}

This provides a distinguisher between random codes and algebraically
structured codes such as generalised Reed Solomon codes~\cite{
W10, CGGOT14}, 
Reed Muller codes~\cite{CB14}, polar codes~\cite{BCDOT16}
some Goppa codes \cite{FGOPT13, COT17}
or algebraic geometry codes~\cite{CMP17}.
For instance, in the case of GRS codes, we have the following result.
\begin{prop}
  Let $n, k, \xv, \yv$ be as in Definition~\ref{def:GRS}. Then,
  \[
  \sqb{\GRS{k}{\xv}{\yv}} = \GRS{2k-1}{\xv}{\yv \star \yv}.
  \]
  In particular, if $k < n/2$, then 
 \[
  \dim \sqb{\GRS{k}{\xv}{\yv}} = 2k-1.
 \]
\end{prop}

Thus, compared to random codes whose square have dimension quadratic in
the dimension of the code, the square of a GRS code has a dimension
which is linear in that of the original code.
This criterion allows to distinguish GRS codes of appropriate dimension from random codes.

\subsection{Punctured and shortened codes}\label{subsec:shortening}

The notions of \textit{puncturing} and \textit{shortening} are
classical ways to build new codes from existing ones. These
constructions will be useful for the attack. We recall here their
definition. Here, for a codeword $\cv \in \F_q^n$, we denote
$(c_1,\ldots,c_n)$ its entries.

\begin{definition}[punctured code]
  Let $\CC \subseteq \F_q^n$ and $\Jind \subseteq \IInt{1}{n}$.
  The {\em puncturing of $\CC$ at $\Jind$} is defined as the code
  \[
  \pu{\CC}{\Jind} \eqdef \{(c_i)_{i \in \IInt{1}{n} \setminus \Jind}\ {\rm s.t.}\
  \cv \in \CC\}.
  \]
\end{definition}

A punctured code can be viewed as the restriction of the codewords to
a subset of code positions.  It will be sometimes more convenient to
view a punctured code in this way. For this reason, we introduce the
following definition.

\begin{definition}[restricted code]
  Let $\CC \subseteq \F_q^n$ and $\Jind \subseteq \IInt{1}{n}$. The
  restriction of $\CC$ to $\Jind$ is defined as the code
 \[
  \re{\CC}{\Jind} \eqdef \{(c_i)_{i \in \Jind}\ {\rm s.t.}\
  \cv \in \CC\} = \pu{\CC}{\IInt{1}{n}\setminus \Jind}.
  \]
\end{definition}

\begin{definition}[shortened code]
  Let $\CC \subseteq \F_q^n$ and $\Jind \subseteq \IInt{1}{n}$.
  The {\em shortening of $\CC$ at $\Jind$} is defined as the code
  \[
  \sh{\CC}{\Jind} \eqdef \pu{\{\cv \in \CC\ {\rm s.t.}\
  \forall i \in \Jind,\ c_i = 0\}}{\Jind}.
  \]
\end{definition}

Shortening a code is equivalent to puncturing the dual code, as
explained by the following proposition.

\begin{prop}[{\cite[Theorem 1.5.7]{HP03}}]
  Let $\CC$ be a linear code over $\F_q^{n}$ and
  $\Jind \subseteq \IInt{1}{n}$. Then,
\[
\sh{\CC^{\perp}}{\Jind} = (\pu{\CC}{\Jind})^{\perp} \text{ and }
(\sh{\CC}{\Jind})^{\perp} = \pu{\CC^{\perp}}{\Jind},
\]
where $\AC^{\perp}$ denotes the dual of the code $\AC$.
\end{prop}

\begin{nota}
  Throughout the document, the indices of the columns (or positions of
  the codewords) will always refer to the indices in the original
  code, although the code has been punctured or shortened. For
  instance, consider a code $\CC$ of length 5 where every word
  $\cv \in \CC$ is indexed $\cv = (c_1, c_2, c_3, c_4, c_5)$. If we
  puncture $\CC$ in $\{1, 3\}$, a codeword $\cv' \in \pu{\CC}{\{1,3\}}$
  will be indexed $(c'_2, c'_4, c'_5)$ and not $(c'_1, c'_2, c'_3)$.
\end{nota}

\section{The RLCE scheme}
\label{sec:scheme}

\subsection{Presentation of the scheme}\label{subsec:presentation}

The RLCE encryption scheme is a code-based cryptosystem, inspired by
the McEliece scheme.  It has been introduced by Y. Wang in \cite{W16}
and a proposal called ``RLCE-KEM'' has recently been submitted as a response
for the NIST's call for post-quantum cryptosystems \cite{W17}.

For a message $\mv \in \F_q^k$, the cipher text is $\cv = \mv \Gm + \ev$ 
where $\ev \in \F_q^{n+w}$ is a random error vector of small weight  $t$
and $\Gm \in \F_q^{k\times(n+w)}$ is a generator matrix defined  as follows, 
for given parameters $n, k$ and $w$.

\begin{enumerate}
\item Let $\xv, \yv \in \F_q^n$ be respectively a support and a multiplier
(as in Definition~\ref{def:GRS}).
\item Let $\Gm_0$ denote a $k \,\times\, n$ generator matrix of the generalised
Reed--Solomon code $\GRS{k}{\xv}{\yv}$ of length $n$ and dimension $k$.
Denote by $g_1, \ldots, g_{n}$ the columns of $\Gm_0$.
\item\label{item:step_random_columns}
  Let 
$r_1, \ldots, r_{w}$ be column vectors chosen uniformly at random in 
$\F_q^k$.
  Denote by $\Gm_1$ the matrix obtained by inserting the random
  columns between GRS columns at the end of $\Gm_0$ as follows:
  \[\Gm_1 \eqdef [g_1, \ldots, g_{n-w}, g_{n-w+1}, r_1,\ldots,
    g_{n},r_{w}] \in \F_q^{k \times (n+w)}.\]
\item Let 
$\Am_1, \ldots, \Am_{w}$ be $2 \times 2$ matrices chosen uniformly at random
in $\mathbf{GL}_2(\F_q)$.
  Let $\Am$ be the block--diagonal non singular matrix
\[
  \Am \eqdef
  \begin{pmatrix}
  \Im_{n-w} & & & (0)\\ & \Am_1 & & \\
  & & \ddots & \\ (0) & & &  \Am_{w}    
  \end{pmatrix}
  \in \F_q^{(n+w)\times(n+w)}.
\]
\item Let $\pi \in \mathfrak{S}_{n+w}$ be a randomly chosen permutation of $\IInt{1}{n+w}$ and $\Pm$ the corresponding $(n+w) \times (n+w)$ permutation matrix. 
\item The public key is the matrix $\Gm \eqdef \Gm_1\Am\Pm$ and
  the private key is $(\xv, \yv, \Am, \Pm)$.
\end{enumerate}

\noindent {\bf Note:} This is a slightly simplified version of the scheme proposed in \cite{W17} without the matrices 
$\Pm_1$ and $\Sm$ of the original description. They are actually not needed and the security of our simplified scheme 
is equivalent to the security of the scheme presented in \cite{W17}.

\subsection{Suggested sets of parameters}

In \cite{W17} the author proposes 2 groups of 3 sets of parameters.
The first group (referred to as \textit{odd ID} parameters) corresponds to parameters such that 
$w \in [0.6(n-k),0.7(n-k)]$, whereas in the second group (\textit{even ID} parameters) the parameters satisfy $w = n-k$.
The parameters of these two groups are listed in Tables~\ref{tab:1stGroup}
and~\ref{tab:2ndGroup}.
  \begin{table}[!h]
  \caption{Set of parameters for the first group : $w \in [0.6(n-k),0.7(n-k)]$.}
  \vspace{.2cm}
  \centering
    \begin{tabular}{|c|c|c|c|c|c|c|c|}
    \hline
    Security level (bits) & Name in \cite{W17} & $n$ & $k$ & $t$ & $w$ & $q$ & Public key size (kB)\\
    \hline
    128 & ID 1& 532 & 376 & 78 & 96 & $2^{10}$ & 118 \\
    \hline
    192 & ID 3 & 846 & 618 & 114 & 144 & $2^{10}$ & 287 \\
    \hline
    256 & ID 5 & 1160 & 700 & 230 & 311 & $2^{11}$ & 742\\
    \hline
  \end{tabular}
  \label{tab:1stGroup}
  \end{table}

\begin{table}[!h]
  \centering
  \caption{Set of parameters for the second group : $w = {n-k}$.}
  \vspace{.2cm}
  \begin{tabular}{|c|c|c|c|c|c|c|c|}
    \hline
    Security level (bits) & Name in \cite{W17} & $n$ & $k$ & $t$ & $w$ & $q$ & Public key size (kB)\\
    \hline
    128 & ID 0 & 630 & 470 & 80 & 160 & $2^{10}$ & 188 \\
    \hline
    192 & ID 2 & 1000 & 764 & 118 & 236 & $2^{10}$ & 450 \\
    \hline
    256 & ID 4 & 1360 & 800 & 280 & 560 & $2^{11}$ & 1232\\
    \hline
  \end{tabular}
  \label{tab:2ndGroup}
\end{table}

The attack of the present paper recovers in polynomial
time any secret key when parameters lie in the first group.

\section{Distinguishing by shortening and squaring}

We will show here that it is possible to distinguish some public keys
from random codes by computing the square of some shortening of the
public code. More precisely, here is our main result.

\begin{thm}\label{thm:main}
  Let $\CC$ be a code over $\Fq$ of length $n+w$ and dimension $k$
  with generator matrix $\Gm$ which is the public key of an RLCE
  scheme that is based on a GRS code of length $n$ and dimension
  $k$. Let $\Jind \subset \Pos$.
  Then,
  \[
  \dim \sqb{\sh{\CC}{\Jind}} \leq \min ( n+w-|\Jind|,\ 
  2(k+w-|\Jind|)-1 ).
  \]
\end{thm}

\subsection{Restriction to the case where $\Pm$ is the identity}

To prove Theorem \ref{thm:main} we can assume that $\Pm$ is the
identity matrix. This is because of the following lemma.

\begin{lemma}\label{lem:square_permut}
  For any permutation $\sigma$ of the code positions $\Pos$ we have
\[ \dim \sqb{\sh{\CC}{\Jind}} = \dim
\sqb{\sh{\CC^\sigma}{\Jind^\sigma}},
\]
where $\CC^\sigma$ is the set of codewords in $\CC$ permuted by
$\sigma$, that is $ \CC^\sigma= \{\cv^\sigma: \cv \in \CC\}$ where
$\cv^\sigma \eqdef (c_{\sigma(i)})_{i \in \Pos}$ and
$\Jind^\sigma \eqdef \{\sigma(i): i \in \Jind\}$.
\end{lemma}

Therefore, for the analysis of the distinguisher,
we can make the following assumption.

\begin{assumption}\label{ass:identity}
  The permutation matrix $\Pm$ is the identity matrix.
\end{assumption}

We will use this assumption several
times the rest of the section, especially to simplify the notation
and define the terminology.
The general case will follow by using
Lemma~\ref{lem:square_permut}.  

\subsection{Analysis of the different kinds of columns}

\subsubsection{Notation and terminology}

Before proving the result, let us introduce some notation and terminology.
The set of positions $\Pos$ splits in a natural way into four sets,
whose definitions are given in the sequel
\begin{equation}
\label{eq:partition}
\Pos = \IGRS^1\cup\IGRS^2\cup\IR\cup\IPR.
\end{equation}


\begin{definition}
  The set of {\em GRS positions of the first type}, denoted $\IGRS^1$, corresponds
  to GRS columns which have not been associated to a random column.
  This set has cardinality $ n-w$ and is given by
  \begin{equation}\label{eq:I1GRS}
    \IGRS^1 \eqdef \{ i \in \Pos \,\vert\, \pi^{-1}(i) \leq n-w\}.
  \end{equation}  
  Under Assumption~\ref{ass:identity}, this becomes:
  \[
    \IGRS^1 \eqdef \IInt{1}{n-w}.
  \]
\end{definition}

This set is called this way, because at a position $i \in \IGRS^1$, any codeword $\vv \in \CC$
  has an entry of the form
  \begin{equation}\label{eq:GRS_pos}
  v_i = y_i  f(x_i).
  \end{equation}
However, there might be other code positions that are of this form, as we will see later.


\begin{definition}
  The set of {\em twin positions}, denoted $\IT$, corresponds
  to columns that result in a mix of a random column and
  a GRS one. This set has cardinality $2w$ and is equal to: 
  \[
    \IT \eqdef \{ i \in \Pos \,\vert\, \pi^{-1}(i) > n-w\}.
  \]
  Under Assumption \ref{ass:identity}, this becomes:
  \[
    \IT \eqdef \IInt{n-w+1}{n+w}.
  \]
\end{definition}

The set $\IT$ can be divided in several subsets as follows.


\begin{definition}
  Each  position in $\IT$ has a unique corresponding \textit{twin} position
  which is the position of the column with which it was mixed.
  For all $s \in \IInt{1}{w}$, $\pi(n-w+2s-1)$ and $\pi(n-w+2s)$
  are twin positions. Under Assumption \ref{ass:identity},
  the positions $n-w+2s-1$ and $n-w+2s$ are twins for all $s$ in $\IInt{1}{w}$.
\end{definition}

For convenience, we introduce the following notation.
 
\begin{nota}
  The twin of a position $i \in \IT$ is denoted by $\tau(i)$.
\end{nota}

To any twin pair $\{i,\tau(i)\}=\{\pi(n-w+2s-1),\pi(n-w+2s)\}$ with $s \in \{1,\dots,w\}$ 
is associated  a unique linear form
  $\psi_s : \F_q[x]_{<k} \rightarrow \Fq$ and a non-singular matrix 
  $\Am_s$ such that for any codeword $\vv \in \CC$, we have
    \begin{equation}\label{eq:PR}
      \begin{array}{rcl}
      v_{i} &=& a_s  y_{j}  f(x_{j}) + c_s  \psi_s(f)\\
      v_{\tau(i)} &=& b_s  y_{j}  f(x_{j}) + d_s  \psi_s(f),
      \end{array}
    \end{equation}
  where $j=n-w+s$ and
  \begin{equation}
\label{eq:Ams}
\begin{pmatrix}
    a_s & b_s \\ c_s & d_s
  \end{pmatrix}= \Am_s.
\end{equation}

The linear form $\psi_s$ is the form whose evaluations provides the
random column added on the right of the $(n-w+s)$--th column during
the construction process of $\Gm$ (see \S~\ref{subsec:presentation},
Step~\ref{item:step_random_columns}).  From \eqref{eq:PR}, we see that
we may obtain more GRS positions: indeed $v_i = a_s y_{j} f(x_{j})$ if
$c_s=0$ or $v_{\tau(i)} = b_s y_{j}f(x_j)$ if $d_s=0$.  On the other
hand if $c_s d_s \neq 0$ the twin pairs are {\em correlated} in the
sense that they behave in a non-trivial way after shortening: Lemma
\ref{lem:derandomize} shows that if one shortens the code in such a
position its twin becomes a GRS position.  We therefore call such a
twin pair a {\em pseudo-random} twin pair and the set of pseudo-random
twin pairs forms what we call the set of {\em pseudo-random} positions.  


\begin{definition}
  The set of {\em pseudo-random positions} (PR in short), denoted
  $\IPR$, is given by
\begin{equation}\label{eq:IPR}
  \IPR \eqdef \bigcup_{s \in \IInt{1}{w} \text{ \textit{s.t.} } c_s d_s \neq 0}
  \{\pi(n-w+2s-1),\pi(n-w+2s)\}.
\end{equation}
Under Assumption \ref{ass:identity}, this becomes:
\begin{equation}\label{eq:IPR_under_assumption}
  \IPR = \bigcup_{s \in \IInt{1}{w} \text{ \textit{s.t.} } c_s d_s \neq 0}
  \{n-w+2s-1,n-w+2s\}.
\end{equation}
\end{definition}

If $c_s d_s=0$, then a twin pair splits into a GRS position of the
second kind and a random position.  The GRS position of the second
kind is $\pi(n-w+2s-1)$ if $c_s=0$ or $\pi(n-w+2s)$ if $d_s=0$ ($c_s$
and $d_s$ can not both be equal to $0$ since $\Am_s$ is invertible).


\begin{definition}
  The set {\em GRS positions of the second kind}, denoted $\IGRS^2$,
  is defined as
\begin{equation}\label{eq:I2GRS}
  \IGRS^2 \eqdef \{\pi(n-w+2s-1) \,|\, c_s = 0\} \cup \{\pi(n-w+2s)
  \,|\, d_s = 0\}.
\end{equation}
Under Assumption \ref{ass:identity}, this becomes:
\begin{equation}\label{eq:I2GRS_under_assumption}
\IGRS^2 = \{n-w+2s-1\,|\,c_s = 0\} \cup \{n-w+2s\,|\,d_s = 0\}.
\end{equation}
\end{definition}


\begin{definition}
The set of {\em random positions}, denoted $\IR$, is defined as
\begin{equation}\label{eq:IR}
\IR \eqdef \{\pi(n-w+2s-1)\,|\,d_s = 0\} \cup \{\pi(n-w+2s)\,|\,c_s = 0\}.
\end{equation}
Under Assumption \ref{ass:identity}, this becomes:
\begin{equation}\label{eq:IR_under_assumption}
\IR = \{n-w+2s-1\,|\,d_s = 0\} \cup \{n-w+2s\,|\,c_s = 0\}.
\end{equation}
\end{definition} 

We also define the {\em GRS positions} to be the GRS positions of the first or the second kind.


\begin{definition}
The set of {\em GRS positions}, denoted $\IGRS$, is defined as
\begin{equation}\label{eq:IGRS}
\IGRS \eqdef \IGRS^1 \cup \IGRS^2.
\end{equation}
\end{definition}
  
\begin{remark}\label{rem:typical_case}
  Note that in the typical case $\IGRS^2$ and $\IR$ are empty
  sets. Indeed, such positions exist only if one of the entries
  (either $c_s$ or $d_s$) of a random non-singular matrix $\mat{A}_s$
  is equal to zero.
\end{remark}

We finish this subsection with a lemma.

\begin{lemma}
  \label{lem:cardIPR}
  $|\IGRS^2| = |\IR|$ and $|\IPR| = 2(w - |\IR|)$.
\end{lemma}

\begin{proof}
  Using (\ref{eq:IPR_under_assumption}),
  (\ref{eq:I2GRS_under_assumption}) and (\ref{eq:IR_under_assumption})
  we see that, under Assumption~\ref{ass:identity},
  \begin{equation}\label{eq:decomp_interval}
  \IInt{n-w+1}{n+w} = \IPR \cup \IGRS^2 \cup \IR 
  \end{equation}
  and the above union is disjoint.
  Next, there is a one-to-one correspondence relating $\IGRS^2$ and $\IR$.
  Indeed, still under Assumption~\ref{ass:identity}, if
  $c_s = 0$ for some $s \in \IInt{1}{w}$, then $n-w+2s-1 \in \IGRS^2$ and
  $n-w+2s \in \IR$ and conversely if $d_s = 0$. This proves that
  $|\IGRS^2| = |\IR|$, which, together with~(\ref{eq:decomp_interval})
  yields the result.
\end{proof}

\subsection{Intermediate results}
  
Before proceeding to the proof of Theorem~\ref{thm:main}, let us state
and prove some intermediate results.  We will start by
Lemmas~\ref{lem:derandomize} and~\ref{lem:subcode}, that will be
useful to prove Proposition~\ref{prop:structure} on the structure of
shortened RLCE codes, by induction on the number of shortened
positions.  This proposition will be the key of the final theorem.
Then, we will prove a general result on modified GRS codes with
additional random columns.

\subsubsection{Two useful lemmas}

The first lemma explains that, after shortening a 
PR position, its twin will behave like a GRS position.
This is actually a crucial lemma that explains
why PR columns in $\Gm$ do not really behave like random
columns after shortening the code at the corresponding position. 

\begin{lem}\label{lem:derandomize}
Let $i$ be a PR position and $\Jind$ a set of positions that neither contains $i$ nor $\tau(i)$.
Let $\CC' \eqdef \sh{\CC}{\Jind}$.
 The position $\tau(i)$ {\em behaves like a GRS position} 
in the code $\sh{\CC'}{\{i\}}$.
  That is, the $\tau(i)$--th column of a generator matrix of $\sh{\CC'}{\{i\}}$
  has entries of the form
  \[
  \tilde{y}_{j}  f(x_{j})
  \]
  for some $j$ in $\IInt{n-w+1}{n}$ and $\tilde{y}_j$ in $\Fq$.
\end{lem}

\begin{proof}
Let us assume that 
$i = n-w+2s-1$ for some $s \in \{1,\dots,w\}$.
The case $i=n-w+2s$ can be proved in a similar way.
    At position $i$, for any $\cv \in \CC'$, from \eqref{eq:PR}, we
    have 
    \[
    c_i = a  y_{j}  f(x_{j}) + c  \psi_s(f),
    \]
    where $j=n-w+s$. 
    By shortening, we restrict our space of polynomials to the subspace
    of polynomials in $\F_q[x]_{<k}$ satisfying $c_i = 0$, \textit{i.e.}
    Since $i$ is a PR position, $c \neq 0$ and therefore
    \[
    \psi_s(f) = - c^{-1}a  y_{j}  f(x_{j}).
    \]
    Therefore, at the twin position $\tau(i)=n-w+2s$ and for any
    $\cv \in \sh{\CC'}{\{i\}}$, we have
    \begin{eqnarray*}
    c_{\tau(i)} &=& b y_j  f(x_j) + d  \psi_j (f)\\
             &=& y_j(b - da c^{-1}) f(x_j).
    \end{eqnarray*}
\end{proof}

\begin{remark}
  This lemma does not hold for a random position, since the proof
  requires that $c \neq 0$. It is precisely because of this that we
  have to make a distinction between twin pairs, \textit{i.e.} pairs
  for which the associated matrix $\Am_s$ is such that
  $c_s d_s \neq 0$ and pairs for which it is not the case.
\end{remark}

This lemma allows us to get some insight on the structure of the
shortened code $\sh{\CC}{\Jind}$. Before giving the relevant statement
let us first recall the following result.

\begin{lemma}
\label{lem:subcode}
Consider a linear code $\code{A}$ over $\Fq$ whose restriction
to a subset $\Jind$ is a subcode of a GRS code over $\Fq$ of dimension
$\kGRS$. Let $i$ be an element of $\Jind$. Then the restriction of
$\sh{\AC}{\{i\}}$ to $\Jind\setminus\{i\}$ is a subcode of a GRS code
of dimension $\kGRS-1$.
\end{lemma}

\begin{proof}
By definition the restriction $\AC'$ to $\Jind$ is a code of the form
\[ \AC' \eqdef \left\{\left(y_j f(x_j)\right)_{j \in \Jind}: f \in
L\right\},
\]
where the $y_j$'s are nonzero elements of $\Fq$, the $x_j$'s are
distinct elements of $\Fq$ and $L$ is a subspace of $\Fq[X]_{<\kGRS}$.
Clearly the restriction $\AC''$ of $\sh{\AC}{\{i\}}$ to
$\Jind\setminus\{i\}$ can be written as
\[ \AC'' = \left\{\left(y_j f(x_j)\right)_{j \in \Jind
\setminus\{i\}}: f \in L,f(x_i)=0\right\}.
\] The polynomials $f(X)$ in $L$ such that $f(x_i)=0$ can be written
as $f(X)=(X-x_i)g(X)$ where $\deg g = \deg f-1$ and $g$ ranges in this
case over a subspace $L'$ of polynomials of degree $< \kGRS-1$.  We
can therefore write
\[ \AC'' = \left\{\left(y_j(x_j-x_i)g(x_j)\right)_{j \in \Jind
\setminus\{i\}}: g \in L'\right\}.
\] This implies our lemma.
\end{proof}

\subsubsection{The key proposition}

Using Lemmas~\ref{lem:derandomize} and~\ref{lem:subcode},
we can prove the following result by
induction. This result is the key proposition for proving Theorem
\ref{thm:main}.

\begin{prop}\label{prop:structure}
Let $\Jind$ be a subset of $\Pos$ and let $\Jind_0, \Jind_1, \Jind_2$ be subsets of $\Jind$ defined as
\begin{itemize}
\item $\Jind_0$ the set of GRS positions (see~(\ref{eq:I1GRS}), (\ref{eq:I2GRS}) and~(\ref{eq:IGRS}) for a definition) of
  $\Jind$, \textit{i.e.}
  \[\Jind_0\eqdef \Jind \cap \IGRS;\]
\item $\Jind_1$ the set of PR positions (see~(\ref{eq:IPR})) of
  $\Jind$ that do not have their twin in $\Jind$, \textit{i.e.}
  \[\Jind_1 \eqdef \{ i \in \Jind \cap \IPR \,\vert\, \tau(i) \not\in
  \Jind\};\]
\item $\Jind_2$ the set of PR positions of $\Jind$ whose twin position
  is also included in $\Jind$, \textit{i.e.}
  \[\Jind_2 \eqdef \{ i \in \Jind \cap \IPR \,\vert\, \tau(i) \in \Jind\}.\]
\end{itemize}

Let $\CC'$ be the restriction of $\sh{\CC}{\Jind}$ to
$(\IGRS\setminus\Jind_0)\cup\tau(\Jind_1)$.  Then, $\CC'$ is a subcode
of a GRS code of length $|\IGRS|-|\Jind_0|+|\Jind_1|$ and dimension
$k-|\Jind_0|-\frac{|\Jind_2|}{2}\cdot$
\end{prop}

\begin{proof}
Let us prove by induction on $\ell = |\Jind|$ that 
$\CC'$ is a subcode of a GRS code 
of length $|\IGRS|-|\Jind_0|+|\Jind_1|$ and dimension
$k-|\Jind_0|-\frac{|\Jind_2|}{2}\cdot$

This statement is clearly true if $\ell=0$, \textit{i.e.} if $\Jind$ is the empty set.
Assume that the result is true for all $\Jind$ up to some size $\ell \geq 0$.
Consider now a set $\Jind$ of size $\ell+1$. 
We can write $\Jind = \Jind' \cup\{i\}$ where $ \Jind'$ is of size $\ell$.

Let $\Jind_0, \Jind_1, \Jind_2$ be subsets of $\Jind$ as defined in the statement and
$\Jind'_0, \Jind'_1, \Jind'_2$ be the subsets of $\Jind'$ obtained by replacing in the statement $\Jind$ by $\Jind'$.
There are now several cases to consider for $i$. 

\begin{itemize}
\item[{\bf Case 1:}] $i \in \Jind_0$. In this case,
  $\Jind_0 = \Jind'_0 \cup \{i\}$, $\Jind_1 = \Jind'_1$ and $\Jind_2= \Jind'_2$.
  We can apply Lemma~\ref{lem:subcode} with $\AC=\sh{\CC}{\Jind'}$ because by the induction hypothesis, 
  its restriction to
  \[
  \Jind'' \eqdef (\IGRS\setminus\Jind'_0)\cup\tau(\Jind'_1)
  \] 
  is a subcode of a GRS code of length $|\IGRS|-|\Jind'_0|+|\Jind'_1|$
  and dimension $k-|\Jind'_0|-\frac{|\Jind'_2|}{2}\cdot$ Therefore the
  restriction of the shortened code
  $\sh{\CC}{\Jind} = \sh{\AC}{\{i\}}$ to
  $\Jind''\setminus\{i\}=(\IGRS\setminus\Jind_0)\cup\tau(\Jind_1)$ is
  a subcode of a GRS code of length $|\IGRS|-|\Jind_0|+|\Jind_1|$ and
  dimension
  $k-|\Jind'_0|-\frac{|\Jind'_2|}{2}-1
  =k-|\Jind_0|-\frac{|\Jind_2|}{2}\cdot$

\item[{\bf Case 2:}] $i \in \Jind_1$. In this case,
  $\Jind_0 = \Jind'_0, \Jind_1 = \Jind'_1 \cup \{i\}$ and
  $\Jind_2= \Jind'_2$.  This implies that $\Jind'$ does not contain
  $i$ nor $\tau(i)$.  We can therefore apply
  Lemma~\ref{lem:derandomize} with $\CC' = \sh{\CC}{\Jind'}$.
  Lemma~\ref{lem:derandomize} states that the position $\tau(i)$
  behaves like a GRS position in $\sh{\CC'}{\{i\}} = \sh{\CC}{\Jind}$.
  By induction hypothesis, the restriction of the code $\CC'$ to
  $(\IGRS\setminus\Jind'_0)\cup\tau(\Jind'_1)$ is a subcode of a GRS
  code of length $|\IGRS|-|\Jind'_0|+|\Jind'_1|$ and dimension
  $k-|\Jind'_0|-\frac{|\Jind'_2|}{2}
  =k-|\Jind_0|-\frac{|\Jind_2|}{2}\cdot$
  Therefore the restriction of $\sh{\CC'}{\{i\}}= \sh{\CC}{\Jind}$ to
  $(\IGRS\setminus\Jind_0)\cup\tau(\Jind_1)=
  (\IGRS\setminus\Jind'_0)\cup\tau(\Jind'_1)\cup\{\tau(i)\}$
  is a subcode of a GRS code of dimension
  $k-|\Jind_0|-\frac{|\Jind_2|}{2}$ and length
  $|\IGRS|-|\Jind'_0|+|\Jind'_1|+1=|\IGRS|-|\Jind_0|+|\Jind_1|$.

\item[{\bf Case 3:}] $i \in \Jind_2$. In this case,
  $\Jind_0 = \Jind'_0, \Jind_1 = \Jind'_1 \setminus \{\tau(i)\}$ and $\Jind_2= \Jind'_2 \cup \{i, \tau(i)\}$.
  In fact, this case can only happen if $\ell \geq 1$ and we will rather consider the induction with respect to the set $\Jind'' = \Jind \setminus \{i, \tau(i)\}$ of size $\ell-1$ and the sets $\Jind''_0, \Jind''_1, \Jind''_2$ such that $\Jind''_0 = \Jind_0, \Jind''_1 = \Jind_1, \Jind''_2 = \Jind_2 \setminus\{i,\tau(i)\}$.
  
By induction hypothesis on $\Jind''$, the restriction of $\CC'' \eqdef \sh{\CC}{\Jind''}$ to 
$(\IGRS\setminus\Jind''_0)\cup\tau(\Jind''_1)$ is a subcode of a GRS code of length 
$|\IGRS|-|\Jind''_0|+|\Jind''_1|=|\IGRS|-|\Jind_0|+|\Jind_1|$ and dimension
$k-|\Jind''_0|-\frac{|\Jind''_2|}{2} =k-|\Jind_0|-\frac{|\Jind_2|}{2}+1$.

Following Assumption~\ref{ass:identity}, we can write without loss of generality that 
$i = n-w+2s-1$ for some $s \in \{1,\dots,w\}$.
The case $i=n-w+2s$ can be proved in a similar way.
Denote $\Am_s = \begin{pmatrix} a & b \\ c & d \end{pmatrix}$ the non-singular matrix and $j = n-w+s$.

For any $\cv \in \CC'$, at positions $i$ and $\tau(i)$ we have 
\begin{align*}
 c_i       & = a  y_{j}  f(x_{j}) + c  \psi_s(f), \\
 c_{\tau(i)} & = b  y_{j}  f(x_{j}) + d  \psi_s(f).  
\end{align*}

Shortening $\CC''$ at $\{i,\tau(i)\}$ has the effect
of requiring to consider only the polynomials $f$ for which
$f(x_j) = \psi_s(f) = 0$. Therefore the restriction of
$\sh{\CC''}{\{i,\tau(i)\}}=\sh{\CC}{\Jind}$ at 
$(\IGRS\setminus\Jind''_0)\cup\tau(\Jind''_1)$ is a subcode of a GRS code of length 
$|\IGRS|-|\Jind_0|+|\Jind_1|$ and dimension 
$k-|\Jind_0|-\frac{|\Jind_2|}{2}+1-1=k-|\Jind_0|-\frac{|\Jind_2|}{2}\cdot$

\item[{\bf Case 4:}] $i \in \IR$.  In this case $\Jind_0 = \Jind'_0,
\Jind_1 = \Jind'_1$ and $\Jind_2= \Jind'_2$. Using the induction
hypothesis yields directly that $\AC=\sh{\CC}{\Jind'}$ is a subcode of
a GRS code of length
$|\IGRS|-|\Jind'_0|+|\Jind'_1|=|\IGRS|-|\Jind_0|+|\Jind_1|$ and
dimension
$k-|\Jind'_0|-\frac{|\Jind'_2|}{2}=k-|\Jind_0|-\frac{|\Jind_2|}{2}\cdot$ This
is also clearly the case for $\sh{\CC}{\Jind}=\sh{\AC}{\{i\}}$.
\end{itemize} This proves that the induction hypothesis also holds for
$|\Jind|=\ell+1$ and finishes the proof of the proposition.
\end{proof}

\subsubsection{A general result on modified GRS codes}

Finally, we need a very general result concerning modified GRS codes
where some arbitrary columns have been joined to the generator matrix.
A very similar lemma is already proved in \cite[Lemma 9]{CGGOT14}.
Its proof is repeated below for convenience and in order to provide further
details about the equality case.

\begin{lemma}\label{lem:square}
  Consider a linear code $\code{A}$ over $\Fq$ with generator matrix
  $\Gm=\begin{pmatrix} \GSCGRS& \GRand
  \end{pmatrix} \Pm$
  of size $k \times (n+r)$ where $\GSCGRS$ is a $k \times n$ generator
  matrix of a subcode of a GRS code of dimension $\kGRS$ over $\Fq$,
  $\GRand$ is an arbitrary matrix in $\Fq^{k \times r}$ and $\Pm$ is
  the permutation matrix of an arbitrary permutation
  $\sigma \in \mathfrak{S}_{n+r}$.  We have
\[
\dim \sq{\code{A}} \leq 2\kGRS-1+r.
\]
Moreover, if the equality holds, then for every $i \in \IInt{n+1}{n+w}$ we have:
\[\dim \pu{\sq{\AC}}{\{\sigma(i)\}} = \dim \sq{\AC} - 1.\]
\end{lemma}

\begin{proof}
  Without loss of generality, we may assume that $\Pm$ is the identity
  matrix since the dimension of the square code is invariant by
  permuting the code positions, as seen in Lemma
  \ref{lem:square_permut}.  Let $\BC$ be the code with generator
  matrix $\begin{pmatrix} \GSCGRS& \zerom{k \times r} \end{pmatrix}$
  where $\zerom{k \times r}$ is the zero matrix of size $k \times r$.
  We also define the code $\BC'$ generated by the generator matrix
  $\begin{pmatrix} \zerom{k \times n} & \GRand \end{pmatrix}$.  We
  obviously have
  \begin{equation*}
    \AC \subseteq \BC + \BC'.
  \end{equation*}
  Therefore 
  \begin{eqnarray*}
    \sqb{\AC} & \subseteq & \sqc{\left(\BC + \BC'\right)}\\
              & \subseteq & \sqc{\BC} + \sqb{\BC'} +  \BC\star \BC' \\
              & \subseteq & \sqc{\BC} + \sqb{\BC'},
  \end{eqnarray*}
  where the last inclusion comes from the fact that $ \BC\star \BC'$
  is the zero subspace since $\BC$ and $\BC'$ have disjoint supports.
  The code $\sqc{\BC}$ has dimension $\leq 2\kGRS-1$ whereas
  $\dim \sqb{\BC'} \leq r$.

  Next, if $\dim \sq{\AC} = 2 \kGRS - 1 + r$, then
  \[
  \sq{\AC} = \sq{\BC} \oplus \sq{(\BC')} \quad {\rm and} \quad
  \dim \sq{(\BC')} = r.
  \]
  Since $\BC'$ has length $r$, this means that $\sq{(\BC')} = \F_q^r$
  and hence, any word of weight $1$ supported by the $r$ rightmost
  positions is contained in $\sq{\AC}$.  Therefore, puncturing this
  position will decrease the dimension.
\end{proof}

\subsection{Proof of the theorem}

We are now ready to prove Theorem \ref{thm:main}.

\begin{proof}[Proof of Theorem~\ref{thm:main}]
  By using Proposition~\ref{prop:structure}, we know that the
  restriction of $\sh{\CC}{\Jind}$ to
  $(\IGRS\setminus\Jind_0)\cup\tau(\Jind_1)$ is a subcode of a GRS
  code of length
  $|\IGRS|-|\Jind_0|+|\Jind_1|=n-w+|\IGRS^2|-|\Jind_0|+|\Jind_1|$ and
  dimension $\kGRS \eqdef k-|\Jind_0|-\frac{|\Jind_2|}{2}$, where:
  \begin{itemize}
  \item
    $\Jind_0\eqdef \IGRS \cap \Jind$;
  \item $\Jind_1$ is
    the set of PR positions of $\Jind$ that do not have their twin in $\Jind$;
  \item $\Jind_2$ is the union of all twin PR positions that are both included in $\Jind$.
  \end{itemize}
  We also denote by $\Jind_3$ the set $\IR \cap \Jind$.
  We can then apply Lemma \ref{lem:square} to $\sh{\CC}{\Jind}$ and derive from it the following 
  upper bound:
  \begin{eqnarray*}
    \dim \sqb{\sh{\CC}{\Jind}} &\leq &2 \kGRS -1 + |\IPR \setminus (\Jind\cup \tau(\Jind_1))| + |\IR \setminus \Jind_3|.
  \end{eqnarray*}
  Next, using Lemma~\ref{lem:cardIPR}, we get
  \begin{eqnarray}
    \dim \sqb{\sh{\CC}{\Jind}}
                               & \leq & 2 \left(k-|\Jind_0|-\frac{|\Jind_2|}{2}\right)-1 + 2\left(w - |\IR| \right) - 2|\Jind_1| -|\Jind_2| + |\IR|-|\Jind_3|\nonumber \\
                               & \leq & 2\left(k +w -|\Jind_0|-|\Jind_1|-|\Jind_2|-|\Jind_3|\right)-1+\left(|\Jind_3|-|\IR|\right)\label{eq:crucial1}\\
                               & \leq & 2\left(k +w -|\Jind|\right)-1.\label{eq:crucial2}
  \end{eqnarray}
  The other upper bound on $\dim \sqb{\sh{\CC}{\Jind}}$ which is
  $\dim \sqb{\sh{\CC}{\Jind}} \leq n+w - |\Jind|$ follows from the
  fact that the dimension of this code is bounded by its length. Putting
  both bounds together yields the theorem.
\end{proof}

\begin{rem}\label{rem:almost_always}
  According to our simulations, the inequality of
  Theorem~\ref{thm:main} is sharp and is attained most of the time
  when $\IR$ is the empty set. When $\IR$ is not the empty set, then
  we may have equality only if we shorten all the positions in $\IR$:
  this is because the right-hand term in \eqref{eq:crucial1} should
  coincide with the right-hand term in \eqref{eq:crucial2}, which is
  equivalent to $|\Jind_3|=|\IR|$.
\end{rem}

\section{Reaching the range of the distinguisher}

For this distinguisher to work we need to shorten the code enough so
that its square does not fill in the ambient space, but not too much
since the square of the shortened code should have a dimension
strictly less than the typical dimension of the square of a random
code given by Proposition~\ref{prop:CCMZ}. Namely, we need to
have:

\begin{equation}\label{eq:dist_interval}
\dim \sqb{\sh{\CC}{\Jind}} < {k + 1 - |\Jind| \choose 2}
\qquad {\rm and}
\qquad
\dim \sqb{\sh{\CC}{\Jind}} < n+w-|\Jind|.
\end{equation}

Thanks to Theorem~\ref{thm:main}, we know that
\eqref{eq:dist_interval} is satisfied as soon as
\begin{equation}\label{eq:plus_fort}
2(k + w - |\Jind|)-1 < {k + 1 - |\Jind| \choose 2}
\qquad {\rm and}
\qquad
2(k + w - |\Jind|)-1 < n+w-|\Jind|.
\end{equation}

We will now find the values of $|\Jind|$ for which both inequalities of
\eqref{eq:plus_fort} are satisfied.

\paragraph{First inequality.} 
To study when the first inequality in \eqref{eq:plus_fort} is verified, 
let us bring in
\[
k' \eqdef k - |\Jind|.
\]

Inequality~(\ref{eq:plus_fort}) becomes
$4k'-2+4w < {k'}^2+k'$, or equivalently
$k'^2-3k'-4w+2 >0$,
which after a resolution leads to
$k' > \frac{3+ \sqrt{16w+1}}{2}\cdot$

Hence, we have: 
\begin{equation}
  \label{eq:upperInd}
  |\Jind| < k - \frac{3+ \sqrt{16w+1}}{2}\cdot
\end{equation}

\paragraph{Second inequality.}
On the other hand, the second inequality in \eqref{eq:plus_fort} 
is equivalent to
\begin{equation}\label{eq:lowerInd}
|\Jind| \geq w + 2k-n.
\end{equation}

\paragraph{Conditions to verify both inequalities.}
Putting inequalities (\ref{eq:upperInd}) and (\ref{eq:lowerInd}) together
gives that $|\Jind|$ should satisfy
\[
w + 2k-n \leq |\Jind| < k - \frac{3+ \sqrt{16w+1}}{2}\cdot
\]
We can therefore find an appropriate $\Jind$ if and only if 
\[
w + 2k-n  < k - \frac{3+ \sqrt{16w+1}}{2},
\]
which is equivalent to
\[
n-k > w+ \frac{3+ \sqrt{16w+1}}{2} = w + O(\sqrt{w}).
\]
In other words, the distinguisher works up to values of $w$ that are close to the second choice 
$n-k=w$.

From now on we set
\begin{eqnarray*}
\ell_{\rm min} &= & w + 2k-n\\
\ell_{\rm max} &= & \left\lceil k - \frac{3+ \sqrt{16w+1}}{2} -1 \right\rceil \cdot
\end{eqnarray*}

\paragraph{Practical results.}\label{ss:practical}
We have run experiments using {\sc Magma} \cite{BCP97} and {\sc Sage}. For
the parameters of Table~\ref{tab:1stGroup}, here are the intervals
of possible values of $|\Jind|$ so that the code $\sh{\CC}{\Jind}^{\star 2}$
has a non generic dimension:
\begin{itemize}
\item ID 1: $n =  532, k = 376, w =  96$,  $|\Jind| \in \IInt{316}{354}$;
\item ID 3: $n =  846, k = 618, w = 144$,  $|\Jind| \in \IInt{534}{592}$;
\item ID 5: $n = 1160, k = 700, w = 311$,  $|\Jind| \in \IInt{551}{663}$.
\end{itemize}
There interval always coincide with the theoretical interval $\IInt{\ell_{\rm min}}{\ell_{\rm max}}$.

\section{The attack}

In this section we will show how to find an equivalent private key
$(\xv, \yv, \Am, \Pm)$ defining the same code. This allows to decode 
and recover the original message like a legitimate user.

We assume that all the matrices $\Am_s = \begin{pmatrix}
  a_s & b_s\\
  c_s & d_s \end{pmatrix}$
appearing in the definition of the scheme in Subsection
\ref{subsec:presentation} are such that $c_s d_s \neq 0$.  We explain
in the appendix how the attack can be changed to take the case
$c_s d_s=0$ into account. Note that this corresponds to a case where
$\IR = \emptyset$ and $\IGRS^2 = \emptyset$, which is the typical case
as noticed in Remark~\ref{rem:typical_case}.

\begin{rem}
  In the present section where we the goal is to recover the permutation,
  we no longer work under Assumption~\ref{ass:identity}.
\end{rem}

\subsection{Outline of the attack}
In summary, the attack works as follows.
\begin{enumerate}
  \item Compute the interval $\IInt{\ell_{\min}}{\ell_{\max}}$
    of the distinguisher and 
    choose $\ell$ in the middle of the distinguisher interval.
    Ensure $\ell < \ell_{\max}$.
  \item For several sets of indices $\Jind \subseteq \Pos$ such that
    $|\Jind| = \ell$, compute $\sh{\CC}{\Jind}$ and identify pairs of
    twin positions contained in $\Pos$. Repeat this process until
    identifying all pairs of twin positions, as detailed in
    \S~\ref{subsec:identify}.
  \item Puncture the twin positions in order to get a GRS code and
    recover its structure using the Sidelnikov Shestakov attack
    \cite{SS92}.
  \item For each pair of twin positions, recover the corresponding
    $2 \times 2$ non-singular matrix $A_i$, as explained in
    \S~\ref{subsec:recover-A}.
  \item Finish to recover the structure of the underlying GRS code.
\end{enumerate}

\subsection{Identifying pairs of twin positions}
\label{subsec:identify}

Let $\Jind \subseteq \Pos$ be such that both $|\Jind|$ and $|\Jind|+1$
are contained in the distinguisher interval.  The idea we use to
identify pairs of twin positions is to compare the dimension of
$\left(\sh{\CC}{\Jind}\right)^{\star 2}$ with the dimension of
$\left(\pu{\sh{\CC}{\Jind}}{\{i\}}\right)^{\star 2}$ for all positions
$i$ in $\Pos \setminus \Jind$.  This yields information on pairs of
twin positions.

\begin{itemize}
\item If $i \in \IGRS$ (see (\ref{eq:I1GRS}), (\ref{eq:I2GRS})
  and (\ref{eq:IGRS}) for the definition), puncturing does not affect the dimension of
  the square code:
  \[
  \dim \left(\sh{\CC}{\Jind}\right)^{\star 2} = \dim
  \left(\pu{\sh{\CC}{\Jind}}{\{i\}}\right)^{\star 2}.
  \]
\item If $i \in \IPR$ (see (\ref{eq:IPR}) for a definition) and
  $\tau(i) \in \Jind$, then according to Lemma~\ref{lem:derandomize},
  the position $i$ is ``derandomised'' in $\sh{\CC}{\Jind}$ and hence
  behaves like a GRS position in the shortened code. Therefore, very
  similarly to the previous case, the dimension does not change.
\item If $i \in \IPR$ and $\tau(i) \not\in \Jind$, in $\sh{\CC}{\Jind}$, the 
  two corresponding columns behave like random ones. 
  Assuming that the inequality of Theorem~\ref{thm:main} is
  an equality, which almost always holds (see
  Remark~\ref{rem:almost_always}),
  then, according to Lemma \ref{lem:square}, puncturing
  $\sh{\CC}{\Jind}^{\star 2}$ at $i$ (resp. $\tau(i)$) reduces its dimension.
  Therefore, 
\[
\dim \left(\pu{\sh{\CC}{\Jind}}{\{i\}}\right)^{\star 2} = \dim
\left(\pu{\sh{\CC}{\Jind}}{\{\tau(i)\}}\right)^{\star 2} = \dim
\left(\sh{\CC}{\Jind}\right)^{\star 2} - 1.
\]
\end{itemize}

This provides a way to identify any position in $\Pos \setminus \Jind$
having a twin which also lies in $\Pos \setminus \Jind$: by searching
zero columns in a parity--check matrix of $\sq{\sh{\CC}{\L}}$, we
obtain the set $\T_{\Jind} \subset \Pos \setminus \Jind$ of even
cardinality of all the positions having their twin in
$\Pos \setminus \Jind$:
\[
\T_{\Jind} \eqdef \bigcup_{\{i,\tau(i)\} \subseteq \Pos \setminus \Jind}\{ i,\tau(i)\}.
\]

As soon as these positions are identified, we can associate each such
position to its twin. This can be done as follows. Take
$i \in \T_{\Jind}$ and consider the code $\sh{\CC}{\Jind \cup \{i\}}$.
The column corresponding to the twin position $\tau(i)$ has been
derandomised and hence will not give a zero column in a parity--check
matrix of $\left(\sh{\CC}{\Jind \cup \{i\}}\right)^{\star 2}$, so
puncturing the corresponding column will not affect the
dimension. 

This process can be iterated by using various shortening
sets $\Jind$ until obtaining $w$ pairs of twin positions. It is
readily seen that considering $O(1)$ such sets is enough to recover
all pairs with very large probability.

\subsection{Recovering the non-randomised part of the code}
\label{subsec:recover-nonrandon}

As soon as all the pairs of twin positions are identified, consider
the code $\pu{\CC}{\IPR}$ punctured at $\IPR$.  Since the randomised
positions have been punctured this code is nothing but a GRS code and,
applying the Sidelnikov Shestakov attack \cite{SS92}, we recover a
pair $\av, \bv$ such that
\[
\pu{\CC}{\IPR} = \GRS{k}{\av}{\bv}.
\]

\subsection{Recovering the remainder of the code and the matrix
$\Am$}
\label{subsec:recover-A}

\subsubsection{Joining a pair of twin positions : the code $\CCi$}

To recover the remaining part of the code we will consider iteratively the
pairs of twin positions. We recall that $\IPR$ corresponds to the set of 
positions having a twin. Let $\{i, \tau(i)\}$ be a pair of twin positions
and consider the code
\[
\CCi \eqdef \re{\CC}{\IGRS \cup \{ i, \tau(i)\}}.
\]

In this code, all but two columns, columns $i$ and $\tau(i)$, are GRS positions. 

For any codeword $\cv \in \CCi$ we have
\begin{equation}\label{eq:cvs}
  \begin{array}{rcl}
  c_{i} &=& a   y_j  f(x_j) + c  \psi_j (f) \\
  c_{\tau(i)} &=& b   y_j  f(x_j) + d  \psi_j(f)
  \end{array}
\end{equation}
for some $j \in \IInt{n-w+1}{n}$,
where $\psi_j$ and
$
\Am = 
  \begin{pmatrix}
    a & b \\ c & d 
  \end{pmatrix}
$
are defined as in \eqref{eq:PR} and \eqref{eq:Ams}.

Note that we do not need to recover exactly $(\xv, \yv, \Am, \Pm)$.
We need to recover a $4$--tuple $(\xv', \yv', \Am', \Pm')$ which
describes the same code. Thus, without loss of generality, after possibly
replacing $a$ by $ay_j$ and $b$ by $by_j$, one can suppose that $y_j = 1$.
Moreover, after possibly replacing $\psi_j$ by $d\psi_j$, one can suppose
that $d = 1$. Recall that in this section we suppose that $cd \neq 0$.

Thanks to these simplifying choices, (\ref{eq:cvs}) becomes
\begin{eqnarray*}
  c_{i} &=& a  f(x_j) + c  \psi_j (f)\\
  c_{\tau(i)} &=& b  f(x_j) + \psi_j(f).
\end{eqnarray*}

\subsubsection{Shortening $\CCi$ at the last position to recover $x_j$}

If we shorten $\CCi$ at the $\tau(i)$-th position,
according to Lemma~\ref{lem:derandomize}, it will
``derandomise'' the $i$-th position (it implies $\psi_j(f)=-bf(x_j)$)
and any $\cv \in \sh{\CCi}{\{\tau(i)\}}$ verifies
\[
c_i = (a - bc)  f(x_j).
\]

Since the support $x_j$ and multiplier $y_j$ are known at all the
positions of $\CCi$ but the two PR ones, for any codeword
$\cv \in \sh{\CCi}{\{\tau(i)\}}$, one can find the polynomial
$f \in \F_q[x]_{<k}$ whose evaluation provides $\cv$. Therefore,
by collecting a basis of codewords in $\sh{\CCi}{\{\tau(i)\}}$ and the
corresponding polynomials, we can recover the values of $x_j$ and $a - bc$.

\subsubsection{Recovering the $2 \times 2$ matrix}

Once we have $x_j$ we need to recover the matrix
\[
  \Am =
\begin{pmatrix}
  a & b \\  c & 1
\end{pmatrix}.
\]
Note that, its determinant $\det \Am = a - bc$ has already been obtained in
the previous section.
First, one can guess $b$ as follows. Let $\Gmi$ be a generator matrix
of $\CCi$. As in the previous section, by interpolation, one can compute
the polynomials $f_1, \ldots, f_k$ whose evaluations provide the rows of
$\Gmi$. Consider the column vector
\[
\vv \eqdef
\begin{pmatrix}
  f_1(x_j) \\ \vdots \\ f_k(x_j)
\end{pmatrix}
\]
and denote by $\vv_i$ and $\vv_{\tau(i)}$ the columns of $\Gmi$ corresponding to positions $c_i$ and $c_{\tau(i)}$:
\[
\vv_i =
\begin{pmatrix}
  af_1(x_j)+c\psi_j(f_1) \\ \vdots \\ af_k(x_j)+c\psi_j(f_k)
\end{pmatrix}
\qquad {\rm and}\qquad
\vv_{\tau(i)} =
\begin{pmatrix}
  bf_1(x_j)+\psi_j(f_1) \\ \vdots \\ bf_k(x_j)+\psi_j(f_k)
\end{pmatrix}.
\]

Next, search $\lambda \in \F_q$ such that $\vv_i - \lambda \vv_{\tau(i)}$
is collinear to $\vv$. This relation of collinearity can be expressed
in terms of cancellation of some $2 \times 2$ determinants which are
polynomials of degree $1$ in $\lambda$. Their common root is nothing
but $c$.

Finally, we can find the pair $(a,b)$ by searching the pairs $(\lambda, \mu)$
such that
\begin{enumerate}[(i)]
\item $\lambda - c \mu = \det \Am$;
\item $\vv_i - \lambda \vv$ and $\vv_{\tau(i)} - \mu \vv$ are collinear.
\end{enumerate}
Here the relation of collinearity will be expressed as the cancellation of 
$2 \times 2$ determinants which are linear combinations of $\lambda, \mu$
and $\lambda \mu$ and elementary elimination process provides us with the value
of the pair $(a,b)$.

\section{Complexity of the attack}

The most expensive part of the attack is the step consisting in
identifying pairs of twin positions. Recall that, from \cite{CGGOT14},
the computation of the square of a code of length $n$ and dimension
$k$ costs $O(k^2n^2)$ operations in $\F_q$. We need to compute the 
square of a code $O(w)$ times, because there are $w$ pairs of twin positions.
Hence this step has a total complexity of $O(wn^2 k^2)$ operations in $\F_q$.
Note that the actual dimension of the shortened codes is significantly
less than $k$ and hence the previous estimate is overestimated.

The cost of the Sidelnikov Shestakov attack is that of a Gaussian
elimination, namely $O(nk^2)$ operations in $\F_q$ which is negligible
compared to the previous step.

The cost of the final part is also negligible compared to the computation of
the squares of shortened codes. This provides an overall complexity in
$O(wn^2k^2)$ operations in $\F_q$.

\section*{Conclusion}
We presented a polynomial time key-recovery attack based on a square
code distinguisher against the public key encryption scheme RLCE. This
attack allows us to break all the so-called {\em odd ID} parameters
suggested in \cite{W17}.  Namely, the attack breaks the parameter sets for
which the number $w$ of random columns was strictly less than $n-k$. Our
analysis suggests that, for this kind of distinguisher by squaring
shortenings of the code, the case $w = n-k$ is the critical one. The
{\em even ID} parameters of \cite{W17}, for which the relation $w = n-k$
always holds, remain out of the reach of our attack.

\section*{Acknowledgements}
The authors are supported by French {\em Agence nationale de la
  recherche} grants ANR-15-CE39-0013-01 {\em Manta} and
ANR-17-CE39-0007 {\em CBCrypt}.  Computer aided calculations have been
performed using softwares {\sc Sage} and {\sc Magma} \cite{BCP97}.

\newpage
\bibliographystyle{abbrv}
\bibliography{codecrypto}
\newpage
\appendix
\section{How to treat the case of degenerate twin positions?}
  \label{s:appendix}

  Recall that a pair of twin positions $i, \tau(i)$ is such that any
  codeword $\cv \in \CC$ has $i$--th and $\tau(i)$--th entries of he form:
  \[
  \cv_i = a y_j f(x_j) + b \psi_j(f) \qquad
  \cv_{\tau(i)} = c y_j f(x_j) + d \psi_j(f).
  \]

  This pair is said to be {\em degenerated} if either $b$ or $d$ is
  zero. In such a situation, some of the steps of the attack cannot be
  applied. In what follows, we explain how this rather rare issue
  can be addressed.

  If either $b$ or $d$ is zero, then one of the positions is
  actually a pure GRS position while the other one is PR but 
  the process explained in the article does not manage to associate
  the two twin columns.

  Suppose w.l.o.g. that $b = 0$.  In the first part if the attack,
  when we collect pairs of twin positions, the position $\tau(i)$ will
  be identified as PR with no twin sister {\em a priori}. To find its
  twin sister, we can proceed as follows. For any GRS position $j$
  replace the $j$--th column $\vv_j$ of a generator matrix $\Gm$ of
  $\CC$ by an arbitrary linear combination of $\vv_j$ and the
  $\tau(i)$--th column, this will ``pseudo--randomise'' this column
  and if the $j$--th column is the twin of the $\tau(i)$--th one, this
  will be detected by the process of shortening, squaring and
  searching zero columns in the parity check matrix.

\end{document}